\documentclass[10pt]{article}
\usepackage{amsmath,amssymb,amsbsy,amsthm}

\let\epsilon\varepsilon
\let\phi\varphi

\def\N{\mathbb N}
\def\R{\mathbb R}

\def\S{\mathcal S}

\def\argmax{\operatorname{argmax}}
\def\as{\text{a.s.}}

\newtheorem{theorem}{Theorem}
\newtheorem{lemma}{Lemma}
\newtheorem{remark}{Remark}

\newtheorem{definition}{Definition}

\begin{document}

\title{Nonparametric Statistical Inference for Ergodic
Processes}
\author{Daniil Ryabko${}^*$, Boris Ryabko${}^\#$}
\date{}
\maketitle
\centerline{${}^*$SequeL, INRIA-Lille Nord Europe, France,  daniil{@}ryabko.net}
\centerline{${}^\#$Institute of Computational Technologies of Siberian Branch of Russian Academy of Science,}
\centerline{Siberian State University of Telecommunications and Informatics,}
\centerline{Novosibirsk, Russia;   boris{@}ryabko.net }

\begin{abstract}
In this work a  method for statistical analysis of time series is proposed,   which is used to obtain 
solutions to some classical problems of mathematical statistics under the only assumption
that the process generating the data is stationary ergodic.
Namely, three problems are considered:   goodness-of-fit (or identity) testing, process classification, and the change point problem. For each
of the problems a test is constructed that is asymptotically
accurate for the case when the data is generated by stationary ergodic  processes. 
The tests are based on empirical estimates of distributional distance.
\end{abstract}

\section{Introduction}
{\bf Overview.}
In this work we consider the problem of statistical analysis of time series, 
when nothing is known about the underlying process generating the data, except
that it is stationary ergodic. There is a vast literature on time series
analysis under various parametric assumptions, and also under such  non-parametric 
assumptions as that the process has a  finite memory or possesses certain mixing rates.
While under these settings most of the problems of statistical analysis 
are clearly solvable and efficient algorithms exist, in the general setting
of stationary ergodic processes it is far less clear what can be done in principle,
which problems of statistical analysis admit a solution and which do not.
In this work we propose a method of statistical analysis of time series,
that allows us to demonstrate that some classical statistical problems
indeed admit a solution under the only assumption that the data is stationary 
ergodic, whereas before solutions only for more restricted cases were known.
The solutions are always constructive, that is, we present asymptotically accurate algorithms 
for each of the considered problems. All the algorithms are based on empirical 
estimates of distributional distance, which is in the core of the suggested approach. We suggest that the proposed approach 
can be applied to other problems of statistical analysis of time series, with the view
of establishing principled positive results, leaving the task of finding
optimal algorithms for each particular problem as a topic for further research.

Here we concentrate on the following three problems:  goodness-of-fit (or identity) testing, process
classification, and  the change point problem. 

\noindent{\bf Goodness-of-fit testing.}
The first problem is the following  problem of hypothesis testing. A stationary ergodic process distribution $\rho$
is known theoretically. Given a data sample, it is required to test whether it was
generated by $\rho$, versus it was generated by any other stationary ergodic distribution that is different from $\rho$ 
(goodness-of-fit, or identity testing). 
The case of i.i.d. or finite-memory processes is widely studied (see e.g.~\cite{cs}); in particular, when $\rho$ has a finite memory~\cite{r1} proposes a test against
any stationary ergodic alternative: a test that can be based on an arbitrary universal code.
It was noted in \cite{shields} that an asymptotically accurate test for the case
of stationary ergodic processes over finite alphabet exists (but no test was proposed).
Here we propose a concrete and simple
asymptotically accurate  goodness-of-fit  test, which demonstrates the proposed
approach: to use empirical distributional distance for hypotheses
testing.
By asymptotically accurate test we mean the following.  First, the Type I error
of the test (or its size) is fixed and is given as a parameter to the test.
That is, given any $\alpha>0$ as an input, under $H_0$ (if the data 
sample was indeed generated by $\rho$) the probability  that the test says ``$H_1$'' is not
greater than $\alpha$. Second, under any hypothesis in $H_1$ (that is, if the distribution generating 
the data is different from $\rho$), the test will say ``$H_0$'' 
not more than a finite number of times, with probability 1. 
 In other words, the Type I error  of the test is fixed and the Type II error
can be made not more than a finite number of times, as the data sample increases, with probability 1 under any stationary ergodic alternative.

A comment on this setting is in order. When the alternative  $H_1$ is less general, e.g. 
distributions that have finite-memory   \cite{gut} or known mixing rates, one typically seeks a test 
that has optimal rates of decrease of probability of Type II error to 0. For our case, 
when the alternative is the set of all stationary ergodic processes, this  rate 
is necessarily non-uniform. In this sense, the property that we establish for our test is the  strongest possible.
Observe  that the notion of consistency that we consider is  stronger than requiring that the test makes only a finite number of errors (either Type I or Type II) with 
probability 1, the setting considered,  for example, in the cases when $H_0$ is composite, 
or for the process classification problem  that we address in this work. 


\noindent{\bf Process classification.} In the next problem that we consider, we again 
have to decide whether a data sample was generated by a process satisfying a hypothesis $H_0$ or
a hypothesis $H_1$. However, here $H_0$ and $H_1$ are not known theoretically, but are represented
by two additional  data samples. 
More precisely, the problem is that of process classification, which can be formulated as follows. We are given
three samples
$X=(X_1,\dots,X_k)$, $Y=(Y_1,\dots,Y_m)$ and $Z=(Z_1,\dots,Z_n)$
generated by stationary ergodic processes with
distributions $\rho_X$, $\rho_Y$ and $\rho_Z$. It is known that $\rho_X\ne \rho_Y$,
while either $\rho_Z=\rho_X$ or
$\rho_Z=\rho_Y$. It is required to test which one is the case.  That is, we have to 
decide whether the sample $Z$ was generated by the same process as the sample $X$ or by the same process 
as the sample $Y$. This problem
for the case of dependent
time series was considered for example in \cite{gut}, where a solution is presented under the finite-memory assumption. It is closely related
to  many important problems in statistics and application areas, such as pattern recognition, classification, etc.
Apparently no  asymptotically accurate procedure for process classification
has been known so far for the general case of  stationary ergodic processes.
Here we propose  a test that converges almost surely to the
correct answer. In other words, the test makes only a finite number of errors with probability 1, with respect to 
any stationary ergodic processes generating the data. Unlike in the previous problem, here we do not explicitly distinguish 
between Type I and Type II error, since the hypotheses are by nature symmetric: $H_0$ is ``$\rho_Z=\rho_X$'' and
$H_1$ is ``$\rho_Z=\rho_Y$''.

\noindent{\bf Change point estimation.}
Finally, we consider the change point problem. It is  another
classical problem, with vast literature
on both parametric (see e.g. \cite{basnik})
and non-parametric (see e.g. \cite{brodar}) methods  for solving
it. 
 In this work we address the case where the data is dependent, its form and the structure of dependence is unknown, 
and marginal distributions before and after the change may be the same. We consider the following (off-line) setting of the problem: a (real-valued) sample
$Z_1,\dots,Z_n$ is given, where $Z_1,\dots,Z_k$ are generated
according to some distribution $\rho_X$ and $Z_{k+1},\dots,Z_n$
are generated according to some distribution $\rho_Y$ which is
different from $\rho_X$. It is known that the distributions
$\rho_X$ and $\rho_Y$ are stationary ergodic, but nothing else is
known about them. 
Most literature on change point 
problem for dependent time series assumes that the marginal distributions before and after the change point
are different, and often also make explicit restrictions on the dependence, such as requirements on mixing rates.
Nonparametric methods used in these cases are typically based on Kolmogorov-Smirnov statistic, Cramer-von Mises statistic, 
or generalizations thereof \cite{brodar, carle, gls}.
The main difference of our results is that we do not assume that the single-dimensional  marginals (or finite-dimensional marginals of any given fixed size)  are different, 
and do not make any assumptions on the structure of dependence. The only assumption is that the (unknown) process distributions before and after the change 
point are stationary ergodic.
Our result is a demonstration of that asymptotically accurate change point estimation is possible in this general setting. 

\noindent{\bf Related problems.} Let us briefly relate the three problems for which we present consistent tests
to other problems of statistical analysis of stationary ergodic time series. First, a closely related problem 
is that of {\em homogeneity testing}. The problem is as follows: given two samples, one has to decide
whether they were generated by the same process distribution or by different ones. While solutions to this
problem exist for i.i.d. data (see for example \cite{bg, ds}, and references therein), for stationary 
ergodic processes (and even for a smaller class of B-processes) a consistent test does not exist, even in the binary-valued case,
as was shown in \cite{disit09}.  This problem is closely related to change point {\em detection} problem: given a single sample,
one has to decide whether there was an abrupt change of distribution somewhere. If we know that there was such a change, 
then we can give an asymptotically consistent estimate for it, as we show here;  however, if it is  not known that the change point exists, nobody
 can construct a consistent change point test, because there is no consistent test for homogeneity.
In other words, we can tell where a change point is, if there is one, but we cannot (in general) tell whether
there is one or not (in the case of stationary ergodic distributions).  
Observe that the process classification problem described above turns out to be
  easier than homogeneity testing: a consistent test exists 
for the former (constructed in this work) but not for the latter.

Other hypothesis testing problems that concern stationary time series include testing
for having a certain memory (i.e. testing the hypothesis ``$k$-order Markov process'' versus 
``stationary ergodic, not $k$-order Markov''), testing for membership to parametric families, and others
\cite{kief, mor2, mor4, r2, r1}. Some recent general results that characterize 
those hypotheses about finitely-valued ergodic processes that can be tested are provided in \cite{d10}. Finally, a related problem is that of {\em prediction} or {\em forecasting} \cite{br,  gml, mor3, br2}. 

In this respect, the results of the present work  clarify which problems can and which cannot be solved, 
when the only  assumption on the data is that it is stationary ergodic.

\noindent{\bf Methodology.} All the tests that we construct are based on empirical 
estimates of the so-called distributional distance. For two processes $\rho_1, \rho_2$
a distributional distance is defined as $\sum_{k=1}^\infty w_k |\rho_1(B_k)-\rho_2(B_k)|$,
where $w_k$ are positive summable real weights, e.g. $w_k=2^{-k}$ and $B_k$ range
over a countable field that generates the sigma-algebra of the underlying probability space.
For example, if we are talking about finite-alphabet processes with the binary alphabet $A=\{0,1\}$, $B_k$ would
range over the set $A^*=\cup_{k\in\N} A^k$; that is, over all tuples $0, 1, 00, 01, 10, 11, 000, 001,\dots$ (of course, we could just as well omit, say, $1$ and $11$); therefore,
the distributional distance in this case is the weighted sum of differences of probabilities of all possible tuples.
In this work we consider real-valued processes,  so $B_k$ have to range through a suitable sequence of intervals, all pairs of such intervals, triples, etc. 
(e.g. we can use a sequence of partitions into cubes of decreasing volume, see the next section for formal definitions). 
Although distributional distance is a natural concept that, for stochastic processes, has been studied for a while \cite{gray},
its empirical estimates have not, to our knowledge,  been used for statistical analysis of time series. 
We argue that this distance is rather natural for this kind of problems, first of all, since it can 
be consistently estimated (unlike, for example, $\bar d$ distance, which cannot \cite{ow} be consistently
estimated for the general case of stationary ergodic processes). Secondly, it is always bounded, unlike  (empirical)
  KL divergence, which is often used for statistical inference for time series (e.g.  \cite{cs, r1, ac, cs2, mor1} and others).
Other approaches to statistical analysis of stationary dependent time series include the use of (universal) codes \cite{kief, r1, r2}. 
Here we first show that distributional distance between stationary ergodic processes can be consistently estimated 
based on sampling, and then apply it to construct   consistent tests for the three problems of statistical analysis described above.

Although empirical estimates of the distributional distance involve taking an infinite sum, in practice it is obvious that 
only a finite number of summands has to be calculated. This is due to the fact that empirical estimates have to be compared
to each other or to theoretically known probabilities, and since the (bounded) summands have (exponentially) decreasing weights, 
the result of the comparison is known after  only  finitely many evaluations. Therefore, the algorithms presented can be applied 
in practice. On the other hand, the main value of the results is in the demonstration of what is possible in principle; finding practically efficient procedures for each of the considered problems is an interesting problem for further research.
A closely related but more practical approach is that of tests based on universal codes \cite{r2,br2}.
\section{Preliminaries}
We are considering (stationary ergodic) processes with the alphabet $A=\R$.
The generalization to $A=\R^d$ is straightforward; moreover, the results can be extended
to the case when $A$ is a complete separable metric space.
We use the symbol $A^*$ for $\cup_{i=1}^\infty A^i$.
Elements of $A^*$ are called words or sequences.
For each $k,l\in\N$, let $B^{k,l}$ be  a partition of the set $\R^k$ into $k$-dimensional cubes of with volume $h_l^k$, such that
$h_l^k\to0$ when $l\to\infty$, for every $k\in\N$. Moreover, define $B^k=\cup_{l\in\N}B^{k,l}$.
Let also $\mathcal B=\cup_{k=1}^\infty B^k$; since this set is countable we can
introduce an enumeration $\mathcal B=\{B_i : i\in\N\}$. The set $\{B_i\times A^\infty: i\in\N\}$
 generates the Borel $\sigma$-algebra on $\R^\infty=A^\infty$.
For a set $B\in \mathcal B$ let $|B|$ be the index $k$ of the set $B^k$  that
$B$ comes from: $|B|= k: B\in B^k$.


For a sequence $X\in A^n$ and a set $B\in \mathcal B$ denote $\nu(X,B)$
the frequency with which the sequence $X$ falls in the set $B$
\begin{equation*} 
\nu(X,B):= \left\{ \begin{array}{rl}  {1\over n-|B|+1}\sum_{i=1}^{n-|B|+1}
I_{\{(X_i,\dots,X_{i+|B|-1})\in B\}} & \text{ if }n\ge |B|, \\
0 & \text{ otherwise}\end{array}\right.
\end{equation*} where $X=(X_1,\dots,X_{n})$.
For example, $$\nu\big((0.5, 1.5, 1.2, 1.4, 2.1),([1.0,2.0]\times[1.0,2.0])\big)=1/2.$$

We use the symbol $\S$ for the set of all stationary ergodic processes on $A^\infty$.
The ergodic theorem (see e.g.~\cite{bil}) implies that  for any process $\rho\in\S$ generating
a sequence $X_1,X_2,\dots$ the frequency of observing a tuple that falls into  each
$B\in\mathcal B$ tends to its
limiting (or a priory) probability a.s.:
$$\nu((X_1,\dots,X_n),B)\rightarrow
\rho((X_1,\dots,X_{|B|})\in B)$$ as $n\rightarrow\infty$.
We will often abbreviate  $\rho((X_1,\dots,X_{|B|})\in B)=:\rho(B)$.

\begin{definition}[distributional distance] 
The  distributional distance  is defined for a pair of processes
$\rho_1,\rho_2$ as follows~\cite{gray}:
\begin{equation}
d(\rho_1,\rho_2)=\sum_{i=1}^\infty w_i |\rho_1(B_i)-\rho_2(B_i)|,
\end{equation}
where $w_i$ are summable positive real weights (e.g. $w_k=2^{-k}$).
\end{definition}
It is easy to see that $d$ is a metric. The reader is referred to \cite{gray} for more information about $d$ and its properties.

\begin{definition}[empirical distributional distance] For $X,Y\in A^*$, define empirical distributional distance $\hat d(X,Y)$ as
\begin{equation}
\hat d(X,Y):=\sum_{i=1}^\infty w_i |\nu(X,B_i)-\nu(Y,B_i)|.
\end{equation}
Similarly, we  can define the empirical distance when only one of the process measures is unknown:
\begin{equation}
\hat d(X,\rho):=\sum_{i=1}^\infty w_i |\nu(X,B_i)-\rho(B_i)|,
\end{equation}
where $\rho\in\S$  and $X\in A^*$.
\end{definition}

The following
lemma will play a key role in establishing the main results.
\begin{lemma}\label{th:dd} Let two samples $X=(X_1,\dots,X_k)$ and
$Y=(Y_1,\dots,Y_m)$ be generated by
stationary ergodic processes $\rho_X$ and $\rho_Y$ respectively. Then
\begin{itemize}
\item[(i)] $\lim_{k,m\rightarrow\infty}\hat d(X,Y)=d(\rho_X,\rho_Y)\ \as$
\item[(ii)] $\lim_{k\rightarrow\infty}\hat d(X,\rho_Y)=d(\rho_X,\rho_Y)\ \as$
\end{itemize}
\end{lemma}
\begin{proof}
For any $\epsilon>0$ we can find such an index $J$ that
$\sum_{i=J}^\infty w_i<\epsilon/2$.
Moreover, for each $j$ we have $\nu((X_1,\dots,X_k),B_j)\rightarrow \rho_X(B_j)$
a.s., so that
$$
 |\nu((X_1,\dots,X_k),B_j) - \rho(B_j)|<\epsilon/(4Jw_j)
$$ from some step $k$ on; define $K_j:=k$.
 Let
$K:=\max_{j<J}K_j$ ($K$ depends
on the realization $X_1,X_2,\dots$).
Define analogously $M$ for the sequence $(Y_1,\dots,Y_m,\dots)$. Thus for $k>K$ and $m>M$ we have
\begin{multline*}
   |\hat d(X,Y) - d(\rho_X,\rho_Y)| =
\\
\left|\sum_{i=1}^\infty
w_i\big(|\nu(X,B_i)-\nu(Y,B_i)| - |\rho_X(B_i)-\rho_Y(B_i)| \big)
\right|\\
   \le \sum_{i=1}^\infty w_i\big(|\nu(X,B_i)-\rho_X(B_i)| +
|\nu(Y,B_i)-\rho_Y(B_i)| \big) \\
    \le \sum_{i=1}^J w_i\big(|\nu(X,B_i)-\rho_X(B_i)| +
|\nu(Y,B_i)-\rho_Y(B_i)| \big) +\epsilon/2\\
   \le \sum_{i=1}^Jw_i(\epsilon/(4Jw_i) + \epsilon/(4Jw_i))
+\epsilon/2 =\epsilon,
\end{multline*}
which proves the first statement. The second statement can be proven analogously.
\end{proof}
\begin{remark}
 While for the proofs the single-index definition of $\rho$ just introduced is more convenient, if the tests are to be computed the 
following definition should be easier to manage (all the statements  below hold for this metric too)
$$
d'(\rho_1,\rho_2):=\sum_{k.l} w_{k,l}\sum_{b\in B^{k,l}}|\rho_1(b)-\rho_2(b)|,
$$
where again the weights $w_{k,l}$ should be summable, e.g. $w_{k,l}:=2^{-(k+l)}$.
\end{remark}

\section{Main results}
\subsection{Goodness-of-fit Test}
For a given stationary ergodic process measure $\rho$ and a sample
$X=(X_1,\dots,X_n)$ we wish to test the hypothesis $H_0$ that the sample was
generated by $\rho$ versus $H_1$ that it was generated by a stationary ergodic
distribution that is  different from $\rho$.
Thus, $H_0=\{\rho\}$ and $H_1=\S\backslash H_0$.

Define the set $D^n_\delta$ as the set of all samples of length  $n$ that are at least $\delta$-far from $\rho$ in empirical distributional distance:
$$
 D^n_\delta:=\{X\in A^n: \hat d(X,\rho)\ge\delta\}.
$$
For each $n$ and each given confidence level $\alpha$ define the critical region $C^n_\alpha$ of the test
as  $C^n_\alpha:=D^n_\gamma$ where
\begin{equation}
  \gamma:=\inf\{\delta: \rho(D_\delta^n)\le\alpha\}.
\end{equation}
The test 
  rejects $H_0$  at confidence level $\alpha$ if $(X_1,\dots,X_n)\in C^n_\alpha$ and accepts it otherwise.
In words, for each sequence we measure the distance between the empirical probabilities (frequencies) and the measure $\rho$ (that is, the theoretical $\rho$-probabilities);
we then take a largest ball (with respect to this distance) around $\rho$ that has $\rho$-probability not greater than $1-\alpha$. The test rejects all sequences outside this ball.

\begin{definition}[Goodness-of-fit test] For each $n\in\N$ and $\alpha\in(0,1)$ the goodness-of-fit test $G^\alpha_n: A^n\rightarrow \{0,1\}$ is 
defined as 
$$
G^\alpha_n(X_1,\dots,X_n):=\left\{ \begin{array}{rl} 1 &\text{ if } (X_1,\dots,X_n)\in C_\alpha^n, \\ 0 & \text{ otherwise. } \end{array}\right.
$$
 
\end{definition}

\begin{theorem}  The  test $G^\alpha_n$ has the following properties.
\begin{itemize}
 \item[(i)] For every $\alpha\in(0,1)$ and every $n\in\N$ the Type~I error of the test  is not greater than $\alpha$:  $\rho(G^\alpha_n=1)\le\alpha$. 
\item[(ii)] For every $\alpha\in(0,1)$ the Type~II error goes to 0 almost surely: for every $\rho'\ne\rho$ we have $\lim_{n\rightarrow\infty}G^\alpha_n =1$  with $\rho'$ probability~1. 
\end{itemize}
\end{theorem}
\begin{proof} The first statement holds by construction. To prove the second statement, let the sample $X$ be generated by $\rho'\in\S$, $\rho'\ne\rho$, and define  $\delta= d(\rho,\rho')/2$.
By Lemma~\ref{th:dd}  we have  $\rho(D^n_\delta)\rightarrow0$,
so that $\rho(D^n_\delta)<\alpha$ from some $n$ on; denote it $n_1$.  Thus, for $n>n_1$ we have  $D^n_\delta\subset C^n_\alpha$.  At the same time,
by Lemma~\ref{th:dd}  we have $\hat d(X,\rho)>\delta$ from some $n$ on, which we denote $n_2(X)$, with  $\rho'$-probability~1. So, for $n>\max\{n_1,n_2(X)\}$ we have $X\in D^n_\delta\subset C^n_\alpha$, 
which proves the statement~{\em (ii)}. 
\end{proof}

\subsection{Process classification}
Let there be given three samples $X=(X_1,\dots,X_k)$,
$Y=(Y_1,\dots,Y_m)$ and $Z=(Z_1,\dots,Z_n)$.
Each sample is generated by a stationary ergodic process $\rho_X$,
$\rho_Y$ and $\rho_Z$ respectively.
Moreover, it is known that either $\rho_Z=\rho_X$ or $\rho_Z=\rho_Y$,
but $\rho_X\ne \rho_Y$.
We wish to construct a test that, based on the finite samples $X, Y$
and $Z$ will tell
whether $\rho_Z=\rho_X$ or $\rho_Z=\rho_Y$.

The test chooses the sample $X$ or $Y$ according to whichever is closer to $Z$ in $\hat d$.
That is, we define the test $G(X,Y,Z)$ as follows. If $\hat d(X,Z)\le \hat d(Y,Z)$ then
the test says that the sample Z is generated by the same
process as the sample X, otherwise it says that the sample Z is generated by the same
process as the sample Y.

\begin{definition}[Process classifier] Define the classifier $L:A^*\times A^*\times A^*\rightarrow\{1,2\}$ as follows
$$
L(X,Y, Z):=\left\{ \begin{array}{rl} 1 &\text{ if } \hat d(X,Z)\le \hat d(Y,Z) \\ 2 & \text{ otherwise, } \end{array}\right.
$$
for $X,Y,Z\in A^*$.
 
\end{definition}

\begin{theorem} The test $L(X,Y,Z)$ makes only a finite number of errors when
$|X|, |Y|$ and $|Z|$ go to infinity, with probability 1:  if $\rho_X=\rho_Z$ then
$L(X,Y,Z)=1$ from some $|X|, |Y|, |Z|$ on with probability~1; otherwise
  $L(X,Y,Z)=2$ from some $|X|, |Y|, |Z|$ on with probability~1.
\end{theorem}
\begin{proof}
 From the fact that $d$ is a metric and from Lemma~\ref{th:dd} we
conclude that $\hat d(X,Z)\rightarrow0$ (with probability~1) if and only if
$\rho_X=\rho_Z$. So, if $\rho_X=\rho_Z$ then by assumption
$\rho_Y\ne\rho_Z$ and $\hat d(X,Z)\rightarrow0$ a.s. while
$$\hat d(Y,Z)\rightarrow d(\rho_Y,\rho_Z)\ne0.$$
Thus in this case $\hat d(Y,Z)>\hat d(X,Z)$ from some $|X|, |Y|, |Z|$ on with probability~1, from which moment we have $L(X,Y,Z)=1$. The opposite case is analogous.
\end{proof}
\subsection{Change point problem}
The sample $Z=(Z_1,\dots,Z_{n})$ consists of two concatenated
parts $X=(X_1,\dots,X_k)$ and $Y=(Y_1,\dots,Y_m)$, where $m=n-k$,
so that $Z_i=X_i$ for $1\le i\le k$ and $Z_{k+j}=Y_{j}$ for $1\le
j\le m$. The samples $X$ and $Y$ are generated independently by
two different stationary ergodic processes with alphabet $A=\R$.
The distributions of the processes  are unknown. The value $k$ is
called the \emph{change point}. It is
assumed that $k$ is linear  in $n$; more precisely, $\alpha n  < k
< \beta n$ for some $0<\alpha\le\beta<1$ from some $n$ on.

It is required to estimate the change point $k$
based on the sample~$Z$.

 For each $t$, $1\le t\le n$, denote  $U^t$  the sample
$(Z_1,\dots,Z_t)$ consisting of the first $t$ elements of the sample
$Z$, and denote
  $V^t$ the remainder $(Z_{t+1},\dots,Z_{n})$.
\begin{definition}[Change point estimator]
Define the change point estimate $\hat k: A^*\rightarrow\N$ as follows:
$$
\hat k(X_1,\dots,X_n):=\argmax_{t\in[\alpha n,n- \beta n]} \hat d(U^t,V^t).
$$
\end{definition}
The following theorem establishes asymptotic consistency of this estimator.

\begin{theorem}\label{th:chp}
 For the estimate $\hat k$ of the change point $k$ we have
$$
|\hat k-k|=o(n)\ \as
$$ where $n$ is the size of the sample, and  when $k,n-k\rightarrow\infty$ in
such a way that $\alpha<{k\over n}<\beta$ for some
$\alpha,\beta\in(0,1)$ from some $n$ on.
\end{theorem}
\begin{proof} To prove the statement, we will show that for every
$\gamma$, $0<\gamma<1$  with probability 1 the inequality $\hat d(U^t,V^t)<\hat d(X,Y)$ holds
for each $t$ such that  $\alpha k\le t<\gamma k$ possibly except for a finite number of times. Thus we will
show that linear $\gamma$-underestimates occur only a finite number of
times, and for
overestimate it is analogous. Fix some $\gamma$, $0<\gamma<1$ and
$\epsilon>0$.  Let $J$ be big enough to have $\sum_{i=J}^\infty
w_i<\epsilon/2$ and
also big enough to have an index $j < J$  for which
$\rho_X(B_j)\ne\rho_Y(B_j)$.
Take $M_\epsilon\in\N$  large enough to have
$|\nu(Y,B_i)-\rho_Y(B_i)|\le\epsilon/2J$ for all $m>M_\epsilon$ and
for each $i$, $1\le i\le J$, and also to have $|B_i|/m<\epsilon/J$ for each $i$, $1\le i\le J$.
This is possible since empirical frequencies converge to the limiting probabilities
a.s. (that is, $M_\epsilon$ depends on the realizations $Y_1,Y_2,\dots$) (cf. the proof of Lemma~\ref{th:dd}).
Find a  $K_\epsilon$ (that depends on $X$) such that for all
$k>K_\epsilon$  and for all $i$, $1\le i\le J$ we 
  have 
\begin{equation}\label{eq:fr}
|\nu(U^t,B_i)-\rho_X(B_i)|\le\epsilon/2J \text{ for each $t \in[\alpha n,\dots,k]$}
\end{equation}
 (this is possible simply because $\alpha n\to \infty$). 
Furthermore, we can select $K_\epsilon$ large enough to have
$|\nu((X_s,X_{s+1},\dots,X_k),B_i)-\rho_X(B_i)|\le\epsilon/2J$ for each
$s\le\gamma k$: this follows from~(\ref{eq:fr}) and the indentity $\nu((X_s,X_{s+1},\dots,X_k)= \frac{k}{k-s}\nu((X_1,\dots,X_k)-\frac{s-1}{k-s}\nu(X_1,\dots,X_{s-1})+o(1)$.

So, for each $s\in[\alpha n,\gamma k]$ 
we  have
\begin{multline*}   \left|\nu(V^s,B_j) - \frac{(1-\gamma)k\rho_X(B_j)+m\rho_Y(B_j)}{(1-\gamma)k+m} \right|\\ \le
\Bigg|\frac{(1-\gamma)k\nu((X_s,\dots,X_k),B_j)+m\nu(Y,B_j)}{(1-\gamma)k+m} - \\
 \frac{(1-\gamma)k\rho_X(B_j)+m\rho_Y(B_j)}{(1-\gamma)k+m} \Bigg|
+ \frac{|B_j|}{m+\gamma k}\ \
 \le
    3\epsilon/J,
\end{multline*}
for $k>K_\epsilon$ and $m>M_\epsilon$ (from the definitions of $K_\epsilon$ and $M_\epsilon$).
Hence
\begin{multline*}
\left|\nu(X,B_j)-\nu(Y,B_j)\right|- \left|\nu(U^s,B_j) -\nu(V^s,B_j)\right|\\
 \ \ \ge \left|\nu(X,B_j)-\nu(Y,B_j)\right| \hfill
\\- \left|\nu(U^s,B_j) -
\frac{(1-\gamma)k\rho_X(B_j)+m\rho_Y(B_j)}{(1-\gamma)k+m} \right|-
3\epsilon/J\\
 \ \ \ge \left|\rho_X(B_j)-\rho_Y(B_j)\right| \hfill
\\ - \left|\rho_X(B_j) -
\frac{(1-\gamma)k\rho_X(B_j)+m\rho_Y(B_j)}{(1-\gamma)k+m} \right|-
4\epsilon/J \\ = \delta_j - 4\epsilon/J,
\end{multline*}
for some $\delta_j$ that depends only on $k/m$ and $\gamma$.
Summing over all $B_i$, $i\in\N$,  we get
$$
\hat d(X,Y)-\hat d(U^s,V^s)\ge w_j\delta_j-5\epsilon,
$$
for all $n$ such that $k>K_\epsilon$ and $m>M_\epsilon$, which is
positive for  small enough~$\epsilon$.
\end{proof}

\section*{Acknowledgements} We are grateful to the anonymous reviewer who suggested 
a simplification of the process metric $d$ (see also Remark~1).
 This work has been supported by French National
Research Agency (ANR), project EXPLO-RA ANR-08-COSI-004 (Daniil Ryabko) and by Russian Foundation for Basic
Research, grant 09-07-00005-a (Boris Ryabko).


\begin{thebibliography}{99}

\bibitem{ac}  R.~Ahlswede,   I.~Csiszar, ``Hypothesis testing with communication constraints,'' {\em IEEE Trans.  Information Theory,} vol.~32  no.~4, pp. 533--542, 1986.


\bibitem{basnik} M. Basseville, I. Nikiforov, {\em Detection of Abrupt
Changes: Theory and Applications.} Prentice Hall, 1993.

\bibitem{bg} G. Biau, L. Gy\"orfi, ``On the asymptotic properties of a nonparametric $L_1$-test of
homogeneity,'' {\em IEEE Trans. Information Theory,} vol.~51, pp.~3965--3973, 2005.

\bibitem{carle} E. Carlstein,  S. Lele,  ``Nonparametric change-point estimation for data
from an ergodic sequence,'' {\em Teor. Veroyatnost. i Primenen.} vol.~38, no.~4 (1993), pp.~910--917; translation in {\em Theory Probab. Appl.} vol.~38 no.~4, pp. 726--733, 1993.


\bibitem{bil} P. Billingsley, {\em Ergodic theory and information.} Wiley,
New York, 1965.

\bibitem{brodar} B. Brodsky, B. Darkhovsky. {\em Nonparametric Methods in
Change-Point Problems.} Kluwer Academic Pablishers, 1993.


\bibitem{cs} I.~Csisz\'ar, P.~Shields, ``Notes on Information Theory and Statistics: A tutorial,'' {\em Foundations and Trends in Communications and Information Theory} 1 (2004), p.~1--111.

\bibitem{cs2} I. Csisz\'ar,    ``Information Theoretic Methods in Probability and Statistics,'' {\em Information Theory Soc. Rev. articles,} 1997. Available: http://www.itsoc.
org/review/frrev.html 

\bibitem{gls} L.~Giraitisa, R.~ Leipusb,  D.~Surgailis, 
``The change-point problem for dependent observations,'' {\em Journal of Statistical Planning and Inference}
vol.~53  no.~3, pp.~297--310, 1996.
 
\bibitem{gray} R. Gray. Probability, {\em Random Processes, and Ergodic
Properties.} Springer Verlag, 1988.

\bibitem{gut} M. Gutman, ``Asymptotically Optimal Classification for
Multiple Tests with Empirically Observed Statistics,'' {\em IEEE Trans.
Information Theory,} vol.~35 no.~2, pp.~402--408, 1989.

\bibitem{kief} J.~C. Kieffer,  ``Strongly consistent code-based identification and order estimation for constrained finite-state model classes,''
{\em IEEE Trans. Inform. Theory}  vol.~39 no.~3, pp. 893--902,  1993.


\bibitem{mor1}  L.~Gy\"orfi,   G.~Morvai,   I.~Vajda, ``Information-theoretic methods in testing the goodness of fit,'' In {\em proceedings 
of IEEE International Symposium on Information Theory,} 2000.

\bibitem{gml}  L. Gy\"orfi, G. Morvai, S. Yakowitz  (1998), ``Limits to consistent on-line forecasting for ergodic time series,''
{\em IEEE Trans. Information Theory} vol. 44 , no. 2, pp. 886--892.




\bibitem{mor3} Morvai G., Weiss B., ``Limitations on intermittent forecasting,'' {\em Statistics and Probability Letters,} vol.~72, pp.~285--290, 2005.

\bibitem{mor2} G. Morvai, B. Weiss, ``On classifying
 processes,'' {\em Bernoulli}, vol.~11, no.~3, pp. 523--532, 2005.

\bibitem{mor4} G. Morvai, B. Weiss, ``On estimating the memory of finitarily Markovian processes,'' {\em Ann. Inst. H.~Poincar\'e Probab. Statist.} vol.~43 pp.~15--30, 2007. 

\bibitem{ow} 
Ornstein, D.~S. and Weiss, B.(1990), ``How Sampling Reveals a Process,''
 {\em Annals of Probability} vol.~18 no.~3, pp. 905--930. 



\bibitem{br} B. Ryabko, ``Prediction of random sequences and universal coding,''
{\em Problems of Information Transmission,} vol.~24, pp. 87--96, 1988.

\bibitem{br2} B. Ryabko, ``Compression-Based Methods for Nonparametric Prediction and Estimation of Some Characteristics of Time Series,''
   {\em IEEE Trans. Information Theory}, Vol. 55, No. 9, 2009, pp. 4309--4315.

\bibitem{r2} B.~Ryabko, J.~Astola, ``Universal codes as a basis for nonparametric testing of serial independence for time series,''
  {\em Journal of Statistical Planning and Inference,} vol.~136  no.~12, pp. 4119-4128, 2006.

\bibitem{r1}  B. Ryabko, J. Astola, A. Gammerman,
  ``Application of Kolmogorov complexity and universal codes to identity testing and
nonparametric testing of serial independence for time series,''
  {\em Theoretical Computer Science,}  vol.~359, 2006, pp. 440--448.

\bibitem{d10} D.~Ryabko, ``Testing composite hypotheses about discrete-valued stationary processes,''
 in {\em Proceedgings of Information Theory Workshop} 2010, Cairo, Egypt, pp. 291-295, 2010.

\bibitem{disit09}
D.~Ryabko,
\newblock ``An impossibility result for process discrimination,''
\newblock in {\em Proc. 2009 IEEE International Symposium on Information
  Theory}, pp. 1734--1738, Seoul, South Korea, 2009. 


\bibitem{we} D. Ryabko, B. Ryabko, ``On hypotheses testing for ergodic processes,'' in {\em Proceedgings of Information Theory Workshop} 2008, Porto, Portugal, pp. 281--283.
\bibitem{ds} D. Ryabko,  J. Schmidhuber,   ``Using Data Compressors to Construct Rank Tests,''
{\em Applied Mathematics Letters}, vol.~22 no.~7, pp.~1029--1032, 2009. 

\bibitem{shields}P.~Shields, ``The Interactions Between Ergodic
Theory and Information Theory,'' {\em IEEE Trans. on Information
Theory,} vol.~44, no.~6 (1998), pp.~2079--2093.
\end{thebibliography}
\end{document}